\documentclass[journal, 10pt, two column]{IEEEtran}
\usepackage{amsfonts}
\usepackage{amssymb}
\usepackage{graphicx}
\usepackage{amsmath}
\usepackage{epsfig}
\usepackage{setspace}
\usepackage{subfigure}
\usepackage{color}
\usepackage[center, small]{caption}
\newtheorem{thm}{Theorem}
\newtheorem{definition}{Definition}
\newtheorem{lemma}{Lemma}
\newtheorem{proposition}{Proposition}

\begin{document}
\title{On Design of Distributed Beamforming for Two-Way Relay Networks}
%
\author{Meng Zeng, Rui Zhang, and Shuguang Cui
\thanks{Meng Zeng and Shuguang Cui are with the Department of Electrical and Computer
Engineering, Texas A\&M University, College Station, TX, 77843. Emails:
\{zengm321, cui\}@tamu.edu.}
\thanks{Rui Zhang is with the Institute for Infocomm Research, A*STAR,
Singapore and the Department of Electrical and Computer Engineering,
National University of Singapore. Email: rzhang@i2r.a-star.edu.sg or
elezhang@nus.edu.sg.}}

\maketitle

\begin{abstract}
We consider a two-way relay network, where two source nodes, S1 and S2,
exchange information through a cluster of relay nodes. The relay nodes
receive the sum signal from S1 and S2 in the first time slot. In the
second time slot, each relay node multiplies its received signal by a
complex coefficient and retransmits the signal to the two source nodes,
which leads to a distributed two-way beamforming system. By applying
the principle of analog network coding, each receiver at S1 and S2
cancels the ``self-interference'' in the received signal from the relay
cluster and decodes the message. This paper studies the 2-dimensional
achievable rate region for such a two-way relay network with
distributed beamforming. With different assumptions of channel
reciprocity between the source-relay and relay-source channels, the
achievable rate region is characterized under two setups. First, with
reciprocal channels, we investigate the achievable rate regions when
the relay cluster is subject to a sum-power constraint or
individual-power constraints. We show that the optimal beamforming
vectors obtained from solving the weighted sum inverse-SNR minimization
(WSISMin) problems are sufficient to characterize the corresponding
achievable rate region. Furthermore, we derive the closed form
solutions for those optimal beamforming vectors and consequently
propose the partially distributed algorithms to implement the optimal
beamforming, where each relay node only needs the local channel
information and one global parameter. Second, with the non-reciprocal
channels, the achievable rate regions are also characterized for both
the sum-power constraint case and the individual-power constraint case.
Although no closed-form solutions are available under this setup, we
present efficient algorithms to compute the optimal beamforming
vectors, which are attained by solving a sequence of semi-definite
programming (SDP) problems after semi-definite relaxation (SDR).

\end{abstract}
\section{Introduction}
\label{sec:intro}Cooperative communication has been extensively studied
in past years, where various cooperative relaying schemes have been
proposed, such as amplify-and-forward (AF) \cite{AF},
decode-and-forward (DF) \cite{DF}, compress-and-forward (CF) \cite{CF},
and coded-cooperation \cite{codecooper}. Among these schemes, due to
its simplicity, the AF-based relaying is of the most practical
interest, where multi-antenna relay beamforming has also been explored
to achieve higher spatial diversity \cite{Khoshnevis08}. In certain
resource constrained networks, such as sensor networks, the node size
is limited such that each node could only mount a single antenna
\cite{VMIMO}. In order to exploit the multi-antenna gain in such
size-limited cases, distributed relay beamforming strategies have been
developed where the relaying nodes cooperate to generate a beam towards
the receiver under sum or individual power constraints
\cite{Nassab08}\cite{NT_BF}.

As an extension to the AF-based one-way relaying scheme, the AF-based
two-way relaying scheme \cite{Rui09} is based on the principle of
analog network coding (ANC) \cite{ANC} to support communications in two
directions. Traditionally, two-way relaying avoids the simultaneous
transmissions of two source terminals, and requires four time-slots to
finish one round of information exchange between them. On the contrary,
the two-way relaying scheme proposed in \cite{ANC} allows the relay to
mix the data and amplify-and-forward it, where the two terminals
exploit the underlying self-interference structure. By doing so, the
amount of required transmission time-slots is reduced from four to two
and the overall network throughput is thus improved. There are several
other works discussing such two-way relay systems. In particular, the
authors in \cite{Rui09} characterized the maximum achievable rate
region for the two-way relay beamforming scheme by assuming a single
relay node equipped with multiple antennas and two source nodes each
equipped with a single antenna. As a counterpart of the work in
\cite{Rui09}, the decode-and-forward two-way relaying has been studied
in \cite{DFtwoway}. The authors in \cite{DistBF} studied the AF-based
two-way relay with distributed beamforming, where the focus is to
minimize the total transmit power of the source nodes and the relay
cluster under a given pair of signal-to-noise ratio (SNR) constraints.

The works on characterizing the rate region of two-way relaying has
also been done in \cite{VazetwowayBF}, where authors considered the
distributed beamforming case. However, all of existing works only
obtained numerical solutions. The missing of closed-form solutions
leads to difficulties in designing efficient algorithms due to the lack
of insight into the structure of the optimal beamforming vectors.
Thereby, in this paper, we try to seek the closed-form solutions for
the optimal beamforming vectors to characterize the maximum achievable
rate region and correspondingly propose efficient distributed
algorithms. Our work differs from the work in \cite{Rui09} from two
main aspects. First, we assume a cluster of single-antenna relay nodes
and consider \emph{distributed} two-way relay beamforming rather than
the multiple-antenna single-relay beamforming. Due to the distributed
feature, we will study the case where each relay node has an individual
power constraint in addition to the case where all relay nodes are
subject to a sum-power constraint. Second, we present closed-form
solutions for the optimal beamforming vectors when we have reciprocal
channels.

The rest of the paper is organized as follows. In Section
\ref{sec:system model}, we introduce the system model and define the
achievable rate region. In Sections \ref{sec: recip} and \ref{sec:
nonrecip}, we characterize the achievable rate regions with two
different assumptions on the channel reciprocity. Sub-optimal schemes
with lower complexity are discussed in Section \ref{sec: subopt}.
Numerical results are presented in Section \ref{sec:results} with
conclusions in Section \ref{sec: conclusion}.

\emph{Notations}: We use uppercase bold letters to denote matrices and
lowercase bold letters to denote vectors. The conjugate, transpose, and
Hermitian transpose are denoted by $(\cdot)^*$, $(\cdot)^{T}$, and
$(\cdot)^{H}$, respectively. The phase of a complex variable $a$ is
denoted as $\angle a$. We use $\text{tr}(\cdot)$ and
$\text{rank}(\cdot)$ to represent the trace and the rank of a matrix,
respectively. A diagonal matrix with the elements of vector
$\mathbf{a}$ as diagonal entries is denoted as
$\text{diag}(\mathbf{a})$. $\mathbf{A}\succeq 0$ means $\mathbf{A}$ is
positive semi-definite, $\mathbf{a}\succeq \mathbf{b}$ means $a_i \geq
b_i$ component-wise, and $\odot$ stands for the Hadamard (elementwise)
multiplication.

\section{System Model}\label{sec:system model}
\begin{figure}
 \begin{centering}
  \includegraphics[width=0.45\textwidth]{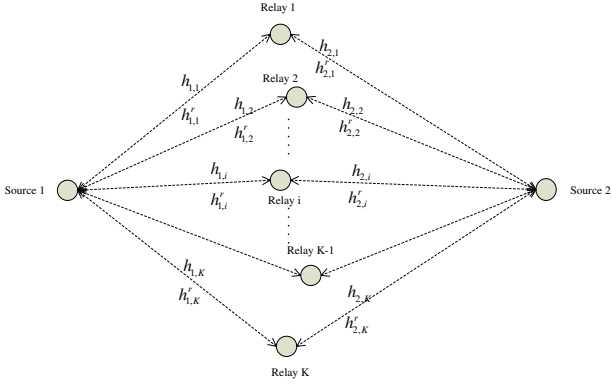}\\
  \caption{System model}\label{fig:sysmodel}
  \end{centering}
\end{figure}
As shown in Fig. \ref{fig:sysmodel}, we consider a distributed two-way
relay system consisting of two source nodes S1 and S2, each with a
single antenna, and a relay cluster with $K$ single-antenna relay nodes
$R_i$'s, $i=1,\cdots,K$. No direct links between S1 and S2 exist. The
forward channels from S1 and S2 to relay node $i$ are denoted as
$h_{1,i}$ and $h_{2,i}$, respectively, while $h_{1,i}^{r}$ and
$h_{2,i}^{r}$ denote the backward channels from relay node $i$ to S1
and S2, respectively.  All the involved channels are assumed to take
complex values and remain constant during one operation period.  In
addition, all channel state information is revealed to S1, S2, and the
design/control center where the beamforming solution is solved.

The two-way relaying takes two consecutive equal-length time-slots to
finish one round of communication between S1 and S2 via the relay
cluster with perfect synchronization assumed among S1, S2, and
$R_i,~i=1,\cdots,K$. In the first time-slot, S1 and S2 transmit their
signals simultaneously to the relay cluster; the $i$-th relay node
receives the mixed signal $t_{i}(n)$, which is expressed as
\begin{equation}\label{eq:rec signal}
    t_{i}(n)=h_{1,i}s_{1}(n)+h_{2,i}s_{2}(n)+v_{i}(n),
\end{equation}
where $s_{1}(n)$ and $s_{2}(n)$ are the transmitted symbols at time
index $n$; and $v_{i}(n)$ is the receiver noise at relay node $i$,
which is assumed to be circularly symmetric complex Gaussian (CSCG)
with zero mean and variance $\sigma_i^2$. In the second time-slot, upon
receiving the mixed signal, relay node $i$ multiplies a complex
coefficient $w_i$ and forwards the signal, which is given as
$u_{i}(n)=w_{i}t_{i}(n)$. At the source node terminals,  S1 and S2
receive the sum signals from all the relay nodes, which are
respectively given as
\begin{eqnarray}
  y_1(n) &=&\sum_{i=1}^{K} h_{1,i}^{r}u_{i}(n)+z_1(n),  \\
  y_2(n) &=&\sum_{i=1}^{K} h_{2,i}^{r}u_{i}(n)+z_2(n),
\end{eqnarray}
where $z_1(n)$ and $z_2(n)$ are the noises at S1 and S2, respectively,
which are assumed to be CSCG with zero mean and variances
$\sigma_{S1}^2$ and $\sigma_{S2}^2$, respectively. Since S1 and S2 know
their own transmitted signals, $s_1(n)$ and $s_2(n)$, respectively,
they could subtract the resulting self-interference terms
$\sum_{i=1}^{K} h_{1,i}^{r}w_{i}h_{1,i}s_1(n)$ and $\sum_{i=1}^{K}
h_{2,i}^{r}w_{i}h_{2,i}s_2(n)$ from the received signals, respectively.
Accordingly, the remaining signals for S1 and S2 are
\begin{eqnarray}
  \tilde{y}_{1}(n) &=& \sum_{i=1}^{K}\left[h_{1,i}^{r}w_ih_{2,i}s_2(n)+h_{1,i}^{r}w_{i}v_{i}(n)\right]+z_1(n), \\
  \tilde{y}_{2}(n) &=&  \sum_{i=1}^{K}\left[h_{2,i}^{r}w_ih_{1,i}s_1(n)+h_{2,i}^{r}w_{i}v_{i}(n)\right]+z_2(n).
\end{eqnarray}
Therefore, for a given $\mathbf{w}=[w_1,\cdots,w_K]^T$ the maximum
achievable rates for the end-to-end link from S2 to S1 and from S1 to
S2 are respectively given as
\begin{eqnarray}
    R_{1}&=&\frac{1}{2}\log_2\left(1+\frac{P_{S2} |\mathbf{f}^{T}_{2}\mathbf{w}|^2}{\sigma_{S1}^2+\mathbf{w}^{H}\mathbf{A}_1\mathbf{w}}\right),\label{eq: r1}\\
    R_{2}&=&\frac{1}{2}\log_2\left(1+\frac{P_{S1}  |\mathbf{ f }^{T}_{1}\mathbf{w}|^2}{\sigma_{S2}^2+\mathbf{w}^{H}\mathbf{A}_2\mathbf{w}}\right),\label{eq: r2}
\end{eqnarray}
where $\mathbf{f}_1=\mathbf{h}_1\odot \mathbf{h}_2^{r}$,
$\mathbf{f}_2=\mathbf{h}_2 \odot \mathbf{h}_{1}^{r}$ with
$\mathbf{h}_i=[h_{i,1},\cdots,h_{i,K}]^T$ and
$\mathbf{h}_{i}^r=[h_{i,1}^r,\cdots,h_{i,K}^r]^T, i=1,2$. In addition,
$\mathbf{A}_1
=\text{diag}[|h_{1,1}^{r}|^2\sigma_{1}^2,\cdots,|h_{1,K}^{r}|^2\sigma_{K}^2]$,
$\mathbf{A}_{2}=\text{diag}[|h_{2,1}^{r}|^2\sigma_{1}^2,\cdots,|h_{2,K}^{r}|^2\sigma_{K}^2]$,
$P_{S1}$ and $P_{S2}$ are the maximum transmit powers at S1 and S2,
respectively, and the factor $1/2$ is due to the use of two orthogonal
time-slots for relaying.

Accordingly, we can define the set of rate pairs achievable by all
feasible beamforming vector $\mathbf{w}$'s as
\begin{equation}\label{eq: achievable rate}
     \mathcal{R}=\mathop {\bigcup} \limits_{\mathbf{w}\in \Omega_{w} } \{(r_1,r_2): r_1\leq R_{1}, r_2\leq
    R_{2}\},
\end{equation}
where the feasible set $\Omega_{w}$ can be defined by either a
sum-power constraint or individual-power constraints. Specifically,
when the sum-power constraint is considered, we have $\Omega_w
=\{\mathbf{w}:p_R(\mathbf{w})\leq P_R\}$, where $P_R$ is a scalar power
limit and $p_R(\mathbf{w})$ is the sum-power of the relay cluster given
the beamforming vector $\mathbf{w}$. When individual-power constraints
are considered, we have $\Omega_w
=\{\mathbf{w}:\mathbf{p}_R(\mathbf{w})\preceq \mathbf{P}_R\}$, where
$\mathbf{p}_R(\mathbf{w})$ is a vector of individual transmit powers,
$\mathbf{P}_R$ is a vector with its elements denoting the power
constraints for individual relay nodes, and $\preceq$ is element-wise.

When time-sharing between different achievable rate pairs is
considered, the achievable rate region is then defined as the convex
hull over the set of $\mathcal{R}$.
\begin{definition}
The achievable rate region $\mathcal{O}$ is the convex hull over the
set of achievable rate pairs $\mathcal{R}$, i.e.,
\begin{equation}
    \mathcal{O}=H_\text{cvx}(\mathcal{R}),
\end{equation}
where $H_\text{cvx}(\cdot)$ is the convex hull operation.
\end{definition}

The goal of this paper is to efficiently characterize the achievable
rate region $\mathcal{O}$. According to different assumptions on the
channel reciprocity between the forward and backward channels, we first
study the reciprocal case, and then study the non-reciprocal case.

\section{Reciprocal Channel Case}\label{sec: recip}
In this section, we assume that the forward channels from each source
node to the relay nodes are reciprocal to the backward channels from
the relay nodes to each corresponding source node, i.e.,
$h_{1,i}=h_{1,i}^r$ and $h_{2,i}=h_{2,i}^r$, for $i=1,\cdots,K$, which
usually holds for a time-division-duplex (TDD) relaying system. In this
case, it is obvious that when $\angle \mathbf{w}=-(\angle
\mathbf{h}_{1}+ \angle \mathbf{h}_{2})$, both rates given by (\ref{eq:
r1}) and (\ref{eq: r2}) are maximized for a given set of $|w_i|$'s.
Thus, we only need to further find the optimal amplitudes for the
elements in $\mathbf{w}$. Let $x_i=|w_i|$ and
$\hat{f}_i=|h_{1,i}||h_{2,i}|$; we rewrite (\ref{eq: r1}) and (\ref{eq:
r2}), respectively, as
\begin{eqnarray}
  R_{1} &=& \frac{1}{2} \log_2\left(1+\frac{P_{S2} \mathbf{ (\hat{f}}^{T}\mathbf{x})^2}{\sigma_{S1}^2+\mathbf{x}^{T}\mathbf{A}_1\mathbf{x}}\right),\label{eq: rate1}\\
  R_{2} &=& \frac{1}{2} \log_2\left(1+\frac{P_{S1} \mathbf{ (\hat{f}}^{T}\mathbf{x})^2}{\sigma_{S2}^2+\mathbf{x}^{T}\mathbf{A}_2\mathbf{x}}\right),\label{eq: rate2}
\end{eqnarray}
where $\mathbf{x}=[|w_1|,\cdots,|w_K|]^T$ and
$\mathbf{\hat{f}}=[|h_{1,1}||h_{2,1}|,\cdots,|h_{1,K}||h_{2,K}|]^T$. In
order to obtain $\mathcal{O}$, we need to characterize the Pareto
boundary of $\mathcal{R}$. A common method is via solving a sequence of
weighted sum-rate maximization (WSRMax) problems \cite{Boyd}, each for
a different non-negative weight vector $(\lambda,1-\lambda), 0\leq
\lambda \leq 1$, as follows
\begin{eqnarray}
  \mathop {\max }\limits_{\mathbf{x}} & & \frac{\lambda}{2}\log_2\left(1+\frac{P_{S2} \mathbf{
  (\hat{f}}^{T}\mathbf{x})^2}{\sigma_{S1}^2+\mathbf{x}^{T}\mathbf{A}_1\mathbf{x}}\right)\nonumber\\
   & &+ \frac{1-\lambda}{2}\log_2\left(1+\frac{P_{S1} \mathbf{ (\hat{f}}^{T}\mathbf{x})^2}{\sigma_{S2}^2+\mathbf{x}^{T}\mathbf{A}_2\mathbf{x}}\right)\label{eq: WSRmax} \\
  \text{s.t.} & & \mathbf{x} \in \Omega_w \label{eq: sum-power WSR}.
\end{eqnarray}
Unfortunately, we cannot derive the closed-form solution for the WSRMax
problem. However, from (\ref{eq: rate1}) and (\ref{eq: rate2}), we see
that the received SNRs at S1 and S2 are
\begin{eqnarray}
  \text{SNR}_1 &=& \frac{P_{S2} \mathbf{ (\hat{f}}^{T}\mathbf{x})^2}{\sigma_{S1}^2+\mathbf{x}^{T}\mathbf{A}_1\mathbf{x}}, \label{eq: snr1}\\
  \text{SNR}_2 &=& \frac{P_{S1} \mathbf{
  (\hat{f}}^{T}\mathbf{x})^2}{\sigma_{S2}^2+\mathbf{x}^{T}\mathbf{A}_2\mathbf{x}}\label{eq: snr2},
\end{eqnarray}
respectively, where their numerators differ by only a scalar constant.
As shown later, for each given weight vector $(\mu,1-\mu), 0\leq
\mu\leq 1$, we could easily find a closed-form solution for the
following weighted sum inverse-SNRs minimization(WSISMin) problem
\begin{eqnarray}
  \mathop {\min }\limits_{\mathbf{x}} & & \mu \frac{\sigma_{S1}^2+\mathbf{x}^{T}\mathbf{A}_1\mathbf{x}}{P_{S2}(\mathbf{ \hat{f}}^{T}\mathbf{x})^2}
  +(1-\mu)\frac{\sigma_{S2}^2+\mathbf{x}^{T}\mathbf{A}_2\mathbf{x}}{P_{S1}(\mathbf{ \hat{f}}^{T}\mathbf{x})^2}\\
  \text{s.t.} & & \mathbf{x} \in \Omega_w.
\end{eqnarray}
Hence, we could quantify the Pareto boundary for the inverse-SNR
region. Based on this observation, together with the fact that there
exists a bijective mapping between an inverse-SNR pair and a rate pair,
we are inspired to probe the question on whether we could construct the
achievable rate region $\mathcal{O}$ from the easily obtainable
inverse-SNR region. In the following, we first introduce some
definitions related to the inverse-SNR region and then show that we
indeed can construct the achievable rate region $\mathcal{O}$ from the
inverse-SNR region, based on the set of closed-form solutions for a
sequence of the WSISMin problems.

\subsection{Characterizing the Achievable Rate Region}\label{subsec: IS definition}
At first, we introduce some definitions.
\begin{definition}\label{def: mapping}
Consider a bijective mapping $\mathcal{U}: (x,y) \mapsto
\left(\frac{1}{2}\log_2(1+1/x),\frac{1}{2}\log_2(1+1/y)\right)$ with
$(x,y)\in R_{++}^{2}$; then the set of achievable inverse-SNR pairs
$\mathcal{I}$ is defined as
\begin{equation*}
    \mathcal{I}=\left\{(t_1,t_2):\mathcal{U}(t_1, t_2) \in \mathcal{R}\right\}.
\end{equation*}
\end{definition}

For regions $\mathcal{R}$ and $\mathcal{I}$, we are particularly
interested in their Pareto boundaries, which are defined as follows.
\begin{definition}
The Pareto boundary of $\mathcal{R}$ is defined as
$\mathcal{P}=\{(r_1,r_2): (r_1,r_2)\in \mathcal{R},
((r_1,r_2)+\mathbf{K})\bigcap\mathcal{R}= (r_1,r_2)\}$, and the Pareto
boundary of $\mathcal{I}$ is defined as $\mathcal{B}=\{(t_1,t_2):
(t_1,t_2)\in \mathcal{I}, ((t_1,t_2)-\mathbf{K})\bigcap\mathcal{I}=
(t_1,t_2)\}$, where $\mathbf{K}=\mathbf{R}_{+}^2$ is a non-negative
cone \cite{Boyd}.
\end{definition}
%

\begin{definition}
Define the points obtained by solving the WSISMin problem with a given
weight vector ($\mu$, $\bar{\mu}$) as
\begin{equation}\label{eq:set mu WS}
    \mathcal{S}(\mu,\mathcal{I})=\{(t_1,t_2): \mathop {\min }\limits_{(t_1, t_2)\in \mathcal{I}} \mu t_1+\bar{\mu} t_2 \}.
\end{equation}
where $\bar{\mu}=1-\mu$.
\end{definition}
\begin{definition}
The set of points that can be obtained from a sequence of WSISMin
problems is given as
\begin{equation}\label{eq:set WS}
    \mathcal{S}(\mathcal{I})=\bigcup \limits_{0\leq \mu \leq 1} \mathcal{S}(\mu,\mathcal{I}).
\end{equation}
\end{definition}

In order to show that we can construct $\mathcal{O}$ from
$\mathcal{S}(\mathcal{I})$, where $\mathcal{O}$ can be obtained by
convex hulling over $\mathcal{P}$,  we need to prove two things: (1) A
point in $\mathcal{B}$ could be mapped to a point in $\mathcal{P}$ and
vice versa, which means $\mathcal{U}(\mathcal{B})=\mathcal{P}$; (2) The
points in $\mathcal{B}$ that cannot be obtained by WSISMin are mapped
to the points in $\mathcal{P}$ that are unnecessary for constructing
$\mathcal{O}$ by convex hulling over $\mathcal{P}$, i.e.,
$\mathcal{U}(\mathcal{B}\setminus \mathcal{S}(\mathcal{I})) \subseteq
\mathcal{P}_{non}$, where $\mathcal{P}_{non}$ denotes the points in
$\mathcal{P}$ that are unnecessary for constructing $\mathcal{O}$ with
convex hulling over $\mathcal{P}$. With the above two statements hold,
it is easy to see that $\mathcal{P}\setminus \mathcal{P}_{non}
\subseteq \mathcal{U}(\mathcal{S}(\mathcal{I}))$, i.e., the points
obtained by WSISMin suffice to construct $\mathcal{O}$.

\begin{proposition}\label{pp: opt points dual}
The Pareto boundary of the inverse-SNR region can be mapped to the
Pareto boundary of the rate region $\mathcal{R}$ by mapping
$\mathcal{U}$ as given in Definition \ref{def: mapping},  and vice
versa, i.e.,
\begin{equation}\label{eq: P mapping B}
    \mathcal{P}=\mathcal{U}(\mathcal{B}).
\end{equation}
\end{proposition}
\begin{proof}
See Appendix in \ref{proof: pp opt points dual}.
\end{proof}

\begin{proposition}\label{pp: unattainable points unnecessary}
The image of $\mathcal{B}\setminus \mathcal{S}(\mathcal{I})$ is not
necessary for constructing the achievable rate region $\mathcal{O}$.
\end{proposition}

\begin{proof}
See Appendix in \ref{proof: pp unattainable points unnecessary}.
\end{proof}

\begin{thm}\label{thm: sufficience to construct O}
The points in $\mathcal{S}(\mathcal{I})$ are sufficient to construct
the achievable rate region $\mathcal{O}$.
\end{thm}
\begin{proof}
Since $\mathcal{P}=\mathcal{U}(\mathcal{B})$ and
$\mathcal{U}(\mathcal{B}\setminus \mathcal{S}(\mathcal{I}))$ in
$\mathcal{P}$ is not necessary for constructing the achievable rate
region $\mathcal{O}$, it is easy to see that
$\mathcal{U}(\mathcal{S}(\mathcal{I}))$ suffices to construct the
achievable rate region $\mathcal{O}$ given that $\mathcal{O}$ is
obtained by convex hulling over $\mathcal{P}$.
\end{proof}
Since we have shown that $\mathcal{S}(\mathcal{I})$ suffices to
construct the achievable rate region $\mathcal{O}$, instead of studying
the problem in (\ref{eq: WSRmax}) and (\ref{eq: sum-power WSR}), we now
study the solutions of the WSISMin problems in the following.
\subsection{Distributed Beamforming under Sum-power Constraint}\label{subsec:sumpower}
In this subsection, we consider the case where the relay cluster has a
sum-power constraint. The total transmit power of the relay cluster is
\begin{eqnarray}\label{eq:sumpower}
    p_{R} &=&\sum_{i=1}^{K}\left(|x_{i}h_{1,i}|^2 P_{S1}+|x_{i}h_{2,i}|^2 P_{S2}+|x_i|^2\sigma_i^2\right)\\
    &=&\mathbf{x}^{H}\mathbf{D}\mathbf{x},
\end{eqnarray}
where we have
$\mathbf{D}=\text{diag}[|h_{1,1}|^2P_{S1}+|h_{2,1}|^2P_{S2}+\sigma_1^2,
\cdots,|h_{1,K}|^2P_{S1}+|h_{2,K}|^2P_{S2}+\sigma_K^2]$. According to
the discussion in the last subsection, to quantify the rate region is
equivalent to seeking the optimal solutions for the WSISMin problems
given the sum-power constraint as follows:
\begin{eqnarray}\label{eq: inverse-SNR OPT}
\mathop {\min }\limits_{\mathbf{x}} & & \mu/\text{SNR}_{1}+\bar{\mu}/\text{SNR}_{2}\\
\text{s.t.} & & \mathbf{x}^{T}\mathbf{D}\mathbf{x}\leq P_R,
\end{eqnarray}
where $0<\mu<1$ is the weight. First of all, the optimal $\mathbf{x}^*$
must satisfy $\mathbf{x}^{*T}\mathbf{D}\mathbf{x}^*= P_R$; otherwise,
we can always scale up $\mathbf{x}$ such that the objective function is
decreased. When $\mu=0~ \text{or}~1$, this problem degrades to find the
optimal beamforming vector for distributed \emph{one-way} relay, which
has been extensively studied
\cite{Khoshnevis08}\cite{Nassab08}\cite{Tang07}.

Given the SNRs from (\ref{eq: snr1}) and (\ref{eq: snr2}) and the fact
of $\mathbf{x}^{*T}\mathbf{D}\mathbf{x}^*= P_R$, the optimal
$\mathbf{x}$ for (\ref{eq: inverse-SNR OPT}) should be the solution of
the following problem:
\begin{equation}
\mathop {\min }\limits_{\mathbf{x}} \frac{\mathbf{x}^{T}[\nu
\mathbf{D}/P_R+\mu/P_{S1} \mathbf{A}_1 +\bar{\mu}/P_{S2} \mathbf{A}_2
]\mathbf{x}}{\mathbf{x}^{T}\mathbf{ \hat{f} } \mathbf{ \hat{f}
}^{T}\mathbf{x}},
\end{equation}
where $\nu=\mu\sigma_{S1}^2 /P_{S2}+\bar{\mu}\sigma_{S2}^2/P_{S1}$. The
above problem is equivalent to
\begin{equation}\label{eq: max sump}
\mathop {\max }\limits_{\mathbf{x}} \frac{\mathbf{x}^{T}\mathbf{
\hat{f} } \mathbf{ \hat{f} }^{T}\mathbf{x}}{\mathbf{x}^{T}[\nu
\mathbf{D}/P_R+\mu/P_{S1} \mathbf{A}_1 +\bar{\mu}/P_{S2} \mathbf{A}_2
]\mathbf{x}},
\end{equation}
where the optimal solution is given as
\begin{equation}\label{eq:opt x}
    \mathbf{x}^*=\xi\mathbf{\Gamma}^{-1}\mathbf{\hat{f}}/\|\mathbf{\Gamma}^{-1}\mathbf{\hat{f}}\|,
\end{equation}
where
\begin{eqnarray}
  \mathbf{\Gamma} &=& \text{diag}\left[\nu \frac{\beta_1}{P_R} +
\eta_1 ,\cdots, \nu \frac{\beta_K}{P_R} +  \eta_K \right], \\
  \beta_i &=& {\sigma_{i}^2 + P_{S1} \left| {h_{1,i} } \right|^2  + P_{S2}
\left| {h_{2,i} } \right|^2 }, \\
  \eta_i &=& \sigma_{i}^2\left(\left|
{h_{1,i} } \right|^2 \mu /P_{S1}  + \left| {h_{2,i} } \right|^2 \bar
\mu /P_{S2}\right),
\end{eqnarray}
and $\xi$ is a scalar such that
$\mathbf{x}^{*T}\mathbf{D}\mathbf{x}^*=P_R$. By searching over all
$\mu$'s, we derived a set of $\mathbf{x}^*$'s and hence we could
compute a set of rate pairs by injecting (\ref{eq:opt x}) into
(\ref{eq: rate1}) and (\ref{eq: rate2}). The achievable rate region
$\mathcal{O}$ is then obtained by convex-hulling over such a set of
rate pairs.

\emph{Partially distributed implementation}: The control center first
decides the appropriate $\mu$ such that S1 and S2 achieve a desirable
rate pair; and it broadcasts $\mu$ and the global constant
$\xi/\|\mathbf{\Gamma}^{-1}\hat{\mathbf{f}}\|$, while $P_{S1}$,
$P_{S2}$, $\sigma_{S1}$, $\sigma_{S2}$ are constant and assumed to be
known at all the relays. Upon receiving the broadcast message from the
control center, each relay node determines the optimal $w_i$ from its
local information $h_{1,i}$ and $h_{2,i}$, which is given as
\begin{equation}
    w_i=\frac{\xi}{\|\mathbf{\Gamma}^{-1}\hat{\mathbf{f}}\|}\frac{|h_{1,i}||h_{2,i}|}{\nu\beta_i/P_R+\eta_i}e^{-j(\angle h_{1,i}+\angle
    h_{2,i})}.
\end{equation}

\subsection{Distributed Beamforming under Individual-Power
Constraints}\label{subsec: indvpower} In the previous subsection, we
assume that the relay cluster has a sum-power constraint. In practice,
each relay may have its own power constraint due to the individual
power supplies. The transmit power at relay $i$ is given as
\begin{equation}\label{eq:sumpower}
    p_{R,i}=|x_i|^2(|h_{1,i}|^2 P_{S1}+|h_{2,i}|^2
    P_{S2}+\sigma_{i}^{2}),
\end{equation}
where $p_{R,i}\leq p_{i}$, with $p_i$ is the maximum allowable power
for relay node $i$. Equivalently, we could set
$p_{R,i}=\alpha_i^2p_{i}$ with $0\leq \alpha_i\leq 1$ as a new design
variable. Correspondingly, the received SNRs can be rewritten as
(\ref{eq:snr1_indv}) and (\ref{eq:snr2_indv})

\begin{figure*}
\begin{eqnarray}
  \text{SNR}_{1} &=& \frac{P_{S2}\left(\sum_{i=1}^{K}|h_{1,i}| |h_{2,i}|\sqrt{\frac{p_{ i}}{\sigma_i^2+P_{S1}|h_{1,i}|^2+P_{S2}|h_{2,i}|^2}}\alpha_i  \right)^2}{\sigma_{S1}^2+\sum_{i=1}^{K}\frac{\sigma_i^2|h_{1,i}|^2 \alpha_{i}^2p_{ i}}{\sigma_i^2+P_{S1}|h_{1,i}|^2+P_{S2}|h_{2,i}|^2}},\label{eq:snr1_indv} \\
  \text{SNR}_{2} &=& \frac{P_{S1}\left(\sum_{i=1}^{K}|h_{2,i}| |h_{1,i}|\sqrt{\frac{p_{ i}}{\sigma_i^2+P_{S1}|h_{1,i}|^2+P_{S2}|h_{2,i}|^2}}\alpha_i  \right)^2}{\sigma_{S2}^2+\sum_{i=1}^{K}\frac{\sigma_i^2|h_{2,i}|^2 \alpha_{i}^2p_{
  i}}{\sigma_i^2+P_{S1}|h_{1,i}|^2+P_{S2}|h_{2,i}|^2}}.\label{eq:snr2_indv}
\end{eqnarray}
\end{figure*}
Let\begin{eqnarray}
     \mathbf{H}_1 &=& \text{diag}[\sigma_1^2p_{1}|h_{1,1}|^2,\cdots,\sigma_K^2p_{K}|h_{1,K}|^2], \\
     \mathbf{H}_2 &=& \text{diag}[\sigma_1^2p_{1}|h_{2,1}|^2,\cdots,\sigma_K^2p_{K}|h_{2,K}|^2], \\
     g_i &=& \sqrt{p_i}|h_{1,i}||h_{2,i}|/\sqrt{\mathbf{D}_{i,i}}.
   \end{eqnarray}
We can recast (\ref{eq:snr1_indv}) and (\ref{eq:snr2_indv}) as
\begin{eqnarray}
  \text{SNR}_{1} &=& \frac{P_{S2}\mathbf{\alpha}^T\mathbf{g}\mathbf{g}^T\mathbf{\alpha}}{\sigma_{S1}^2+\alpha^T\mathbf{H}_1\mathbf{D}^{-1}\mathbf{\alpha}}, \\
  \text{SNR}_{2} &=&  \frac{P_{S1}\mathbf{\alpha}^T\mathbf{g}\mathbf{g}^T\mathbf{\alpha}}{\sigma_{S2}^2+\alpha^T\mathbf{H}_2\mathbf{D}^{-1}\mathbf{\alpha}},
\end{eqnarray}
respectively, where $\mathbf{0}\preceq\mathbf{\alpha}\preceq
\mathbf{1}$. The WSISMin problem for the individual-power constraint
case is now given as:
\begin{equation}\label{eq:  min inversed SNR}
\mathop{\min}\limits_{\mathbf{0}\preceq\mathbf{\alpha}\preceq\mathbf{1}}
\frac{\nu+\mathbf{\alpha}^T(\mathbf{H}_1\mathbf{D}^{-1}\mu/P_{S2}+\mathbf{H}_2\mathbf{D}^{-1}\bar{\mu}/P_{S1})\mathbf{\alpha}}{\mathbf{\alpha}^T\mathbf{g}\mathbf{g}^T\mathbf{\alpha}},
\end{equation}
which is equivalent to solve
\begin{equation}\label{eq:  max harmonic SNR}
\mathop{\max}\limits_{\mathbf{0}\preceq\mathbf{\alpha}\preceq\mathbf{1}}
\frac{\mathbf{\alpha}^T\mathbf{g}\mathbf{g}^T\mathbf{\alpha}}{\nu+\mathbf{\alpha}^T(\mathbf{H}_1\mathbf{D}^{-1}\mu/P_{S2}+\mathbf{H}_2\mathbf{D}^{-1}\bar{\mu}/P_{S1})\mathbf{\alpha}}.
\end{equation}
For notation simplicity, let
\begin{eqnarray}
    \Psi&=&\left[(\mathbf{H}_1\mu/P_{S2}+\mathbf{H}_2\bar{\mu}/P_{S1})\mathbf{D}^{-1}/\nu\right]^{1/2},\\
     \mathbf{\tilde{g}}&=&\mathbf{g}/\sqrt{\nu},
\end{eqnarray}
where $\Psi$ is diagonal with its diagonal elements denoted as
$\psi_i$, $i=1,\cdots,K$. Then the above problem becomes
\begin{equation}\label{eq: max harmonic SNR}
\mathop{\max}\limits_{\mathbf{0}\preceq\mathbf{\alpha}\preceq\mathbf{1}}
\frac{\langle \mathbf{\tilde{g}},
\mathbf{\alpha}\rangle^2}{1+\|\mathbf{\Psi} \mathbf{\alpha}\|^2}.
\end{equation}
For each given $\mu$, (\ref{eq: max harmonic SNR}) can be solved
analytically by following the results in \cite{NT_BF}. Before we
present the solution, we first define $\phi_i  =  \tilde{g}_i/\psi_i^2$
for $i=1,\cdots, K$ and $\phi_{K+1}=0$. Then we sort $\phi_i$ as
$\phi_{\tau_1}\geq\phi_{\tau_2}\geq\cdots\geq\phi_{\tau_K}\geq
\phi_{\tau_{K+1}}$. Moreover, let $\lambda_k
=\frac{1+\sum_{m=1}^{k}\psi_{\tau_{m}}^2}{\sum_{m=1}^{k}\tilde{g}_{\tau_m}}$
and define the $j$-th element of the vector $\mathbf{\alpha}^{(k)}$ as
\begin{equation}\label{eq: opt alpha}
    \alpha_j^{(k)}  = \left\{ {\begin{array}{*{20}c}
   {1,} & {j = \tau _1 , \cdots \tau _k }  \\
   {\lambda _k \phi _j ,} & {j = \tau _{k + 1} , \cdots \tau _K }  \\
\end{array}} \right..
\end{equation}
Then the solution for (\ref{eq: max harmonic SNR})  is given by
following theorem.
\begin{thm}
The solution of (\ref{eq: max harmonic SNR}) is
$\mathbf{\alpha}^{(k^*)}$ given by (\ref{eq: opt alpha}), where $k^*$
is the smallest $k$ such that $\lambda_k<\phi_{\tau_{k+1}}^{-1}$.
\end{thm}
\begin{proof}
This result directly follows the results in \cite{NT_BF}.
\end{proof}

\emph{Partially distributed implementation:} Besides the value of
$\mu$, the control center only needs to broadcast $\lambda_{k^*}$ at
each operation period. Each relay node then determines $\phi_i$ with
its local information. If $\phi_i^{-1}\leq \lambda_{k^*}$, the relay
node transmits at its maximum power. Otherwise, it transmits with power
$(\lambda_{k^*}\phi_i)^2 p_i$, i.e., the optimal
$w_i=\alpha_{i}^{(k^*)} \sqrt{p_i} e^{-j(\angle h_{1,i}+\angle
h_{2,i})}$, where $\alpha_{i}^{(k^*)}$ is given in (\ref{eq: opt
alpha}). From the solutions, we see that in general some relay nodes
may not transmit with maximum transmit power.

\section{Non-reciprocal Channel Case}\label{sec: nonrecip}
In the last section, we have discussed the case where the uplink and
downlink channels are reciprocal. In this section, we discuss the case
where the uplink and downlink channels are non-reciprocal, which may be
the result of deploying frequency-division-duplex (FDD) system.

Due to the lack of channel reciprocity, the approach taken in the last
section does not apply here. In order to characterize the boundary of
the region $\mathcal{R}$, as we discussed before a commonly used method
is to solve the following
\begin{eqnarray}\label{eq: non-rep-opt-formulation}
    \mathop {\max }\limits_\mathbf{w} {\rm  } & & \frac{\lambda}{2} \log_2\left(1+\frac{P_{S2} |\mathbf{ f
    }^{T}_{2}\mathbf{w}|^2}{\sigma_{S1}^2+\mathbf{w}^{H}\mathbf{A}_1\mathbf{w}}\right)\nonumber \\
    & &+\frac{ 1-\lambda }{2}\log_2\left(1+\frac{P_{S1} |\mathbf{ f }^{T}_{1}\mathbf{w}|^2}{\sigma_{S2}^2+\mathbf{w}^{H}\mathbf{A}_2\mathbf{w}}\right) \\
  \text{s.t.} & & \mathbf{w}\in \Omega_w,
\end{eqnarray}
for each given weight vector $(\lambda, 1-\lambda)$. However, the above
problem is non-convex since the objective function is not a concave
function. To efficiently quantify the rate region, here we resort to an
alternative method called the \emph{rate-profile} method \cite{Rui09},
formulated as
\begin{eqnarray}
    \mathop {\max }\limits_{\mathbf{w}, R_{sum} }{\rm  } & & R_{sum}\\
\text{s.t.} & & \frac{1}{2}\log_2\left(1+\frac{P_{S2}|\mathbf{f}^{T}_{2}\mathbf{w}|^2}{\sigma_{S1}^2+\mathbf{w}^{H}\mathbf{A}_1\mathbf{w}}\right)\geq \kappa R_{sum}, \\
   & &\frac{1}{2}\log_2\left(1+\frac{P_{S1} |\mathbf{f}^{T}_{1}\mathbf{w}|^2}{\sigma_{S2}^2+\mathbf{w}^{H}\mathbf{A}_2\mathbf{w}}\right) \geq \bar{\kappa} R_{sum}, \\
   & &\mathbf{w}\in \Omega_w,
\end{eqnarray}
where $R_{sum}$ is the sum rate given a rate profile vector
$[\kappa,\bar{\kappa}]$ with $0\leq \kappa \leq 1$ and
$\bar{\kappa}=1-\kappa$. Let
$\mathbf{F}_1=\mathbf{f}_1^*\mathbf{f}_1^T$,
$\mathbf{F}_2=\mathbf{f}_2^*\mathbf{f}_2^T$, and
$\mathbf{X}=\mathbf{w}\mathbf{w}^{H}$. The above problem is equivalent
to
\begin{eqnarray}
    \mathop {\max }\limits_{\mathbf{X}, R_{sum} } {\rm  } & & R_{sum} \label{eq: general SDP}\\
    \text{s.t.} & &\frac{1}{2}\log_2\left(1+\frac{P_{S2}\text{tr}(\mathbf{F}_2 \mathbf{X})}{\sigma_{S1}^2+\text{tr}(\mathbf{A_1} \mathbf{X})}\right)\geq \kappa R_{sum},  \\
   & &\frac{1}{2}\log_2\left(1+\frac{P_{S1} \text{tr}(\mathbf{F}_1 \mathbf{X})}{\sigma_{S2}^2+\text{tr}(\mathbf{A_2} \mathbf{X})}\right) \geq \bar{\kappa} R_{sum}, \\
   & &\mathbf{X}\in \Omega_X, \label{eq: power const X}\\
    & & \mathbf{X} \succeq 0,  \\
   & &\textrm{rank}(\mathbf{X})=1 ,
\end{eqnarray}
where the last constraint $\textrm{rank}(\mathbf{X})=1$ comes from the
fact $\mathbf{X}=\mathbf{w}\mathbf{w}^H$, and $\Omega_X=\{\mathbf{X}:
\mathbf{X}=\mathbf{w}\mathbf{w}^H, \mathbf{w}\in \Omega_w\}$ and
$\Omega_w$ is defined after (\ref{eq: achievable rate}). According to
different assumptions on the power constraint, the above problem can be
further converted into different semi-definite programming (SDP)
problems after semi-definite relaxation (SDR).
\subsection{Sum-power Constrained Case}\label{subsec:sumpower} In this
subsection, we assume that the relay cluster operates under a sum-power
constraint $P_{R}$. Given the sum-power constraint, the power
constraint in (\ref{eq: power const X}) can be replaced by
$\textrm{tr}(\mathbf{D} \mathbf{X})\leq P_R$, where
$\mathbf{D}=\text{diag}[|h_{1,1}|^2P_{S1}+|h_{2,1}|^2P_{S2}+\sigma_1^2,\cdots,|h_{1,K}|^2P_{S1}+|h_{2,K}|^2P_{S2}+\sigma_K^2]$.
Since the rank-one constraint is not convex, the problem is still not a
convex problem and hence may not be efficiently solvable. To address
this issue, let us first remove the rank-one constraint and consider
the following relay power minimization problem for given set of
$\kappa$ and $R_{sum}=r$:
\begin{eqnarray}
  \mathop {\min }\limits_{\mathbf{X}}  & &  \textrm{tr}(\mathbf{D} \mathbf{X}) \\
  \text{s.t.} & &\frac{1}{2}\log_2\left(1+\frac{P_{S2}\text{tr}(\mathbf{F}_2 \mathbf{X})}{\sigma_{S1}^2+\text{tr}(\mathbf{A_1} \mathbf{X})}\right)\geq \kappa r, \\
   & &\frac{1}{2}\log_2\left(1+\frac{P_{S1} \text{tr}(\mathbf{F}_1 \mathbf{X})}{\sigma_{S2}^2+\text{tr}(\mathbf{A_2} \mathbf{X})}\right) \geq \bar{\kappa} r, \\
  & & \mathbf{X} \succeq 0,
\end{eqnarray}
which is equivalent to
\begin{eqnarray}
  \mathop {\min }\limits_{\mathbf{X}}  & &  \textrm{tr}(\mathbf{D} \mathbf{X}) \label{eq: orig power min}\\
  \text{s.t.} & &\frac{P_{S2}\text{tr}(\mathbf{F}_2 \mathbf{X})}{\sigma_{S1}^2+\text{tr}(\mathbf{A_1} \mathbf{X})}\geq \gamma_1,  \\
   & &\frac{P_{S1} \text{tr}(\mathbf{F}_1 \mathbf{X})}{\sigma_{S2}^2+\text{tr}(\mathbf{A_2} \mathbf{X})} \geq \gamma_2,  \\
  & & \mathbf{X} \succeq 0,
\end{eqnarray}
where $\gamma_1=2^{2\kappa r}-1$, $\gamma_2=2^{2\bar{\kappa} r}-1$, and
they can be considered as the SNR constraints for S1 and S2,
respectively. Since
$\sigma_{S1}^2+\text{tr}(\mathbf{A}_1\mathbf{X})\geq 0$ and
$\sigma_{S2}^2+\text{tr}(\mathbf{A}_2\mathbf{X}) \geq 0$, we could
rewrite the above problem as following SDP problem:
  \begin{eqnarray}
    \mathop {\min }\limits_{\mathbf{X}}  & &  \textrm{tr}(\mathbf{D} \mathbf{X}) \label{eq: power min}\\
    \textrm{s.t.} & & \textrm{tr}[(P_{S2}\mathbf{F}_2-\gamma_1\mathbf{A}_1)\mathbf{X}] \geq \gamma_1 \sigma_1^2,\label{eq: quadratic con 1}\\
     & & \textrm{tr}[(P_{S1}\mathbf{F}_1-\gamma_2\mathbf{A}_2)\mathbf{X}] \geq \gamma_2 \sigma_2^2, \label{eq: quadratic con 2}\\
    & & \mathbf{X} \succeq 0 \label{eq: quadratic con 3}.
   \end{eqnarray}
Denote the optimal value of the above problem as $p_{R}^{*}$, which is
the minimum sum-power required by the relay cluster to support the
target SNRs $\gamma_1$ and $\gamma_2$ for S1 and S2, respectively. If
$p_{R}^{*} \leq P_{R}$, then $(\gamma_1, \gamma_2)$ must be an
achievable SNR pair. Otherwise, $\gamma_1$ and $\gamma_2$ are not
achievable. Based on this observation, we propose the following
bi-section algorithm such that the problem (\ref{eq: general SDP})
without rank-one constraint can be solved by solving a sequence of
convex power feasibility problems, with the assumption that
we know an upper bound for $R_{sum}$, denoted as $r_{max}$.\\
\underline{Algorithm 1:}
\begin{itemize}
  \item Initialize $r_{low}$ = 0, $r_{up}=r_{max}$.
  \item Repeat
  \begin{enumerate}
    \item Set $r \leftarrow \frac{1}{2}(r_{low}+r_{up})$.
    \item Solve problem (\ref{eq: power min})-(\ref{eq: quadratic con 3}) with the given $r$.
    \item Update $r$ with the bi-section method \cite{Boyd}: If
    $p_{R}^{*} \leq P_{R}$, set $r_{low}= r$; otherwise, $r_{up}= r$.
  \end{enumerate}
  \item Until $r_{up}-r_{low}<\epsilon$, where $\epsilon$ is a
  small positive accuracy parameter.
\end{itemize}

The rate upper bound $r_{max}$ can be derived as follows. We first
decouple the two-way relay channel into two one-way relay channels and
obtain a rate for each one-way relay channel. Denote the larger rate as
$\tilde{r}$. Then $r_{max}$ can be set as $2\tilde{r}$. The one-way
distributed relay beamforming with sum-power constraint is
well-studied, and the rate can be derived from the results in
\cite{Nassab08}.
\subsubsection{Rank-one solution}
The resulting optimal solution $\mathbf{X}_{opt}$ obtained from
Algorithm 1 may not be of rank-one due to the SDP relaxation, which
means that $\mathbf{X}_{opt}$ may not lead to an optimal beamforming
vector $\mathbf{w}$. However, since there are only two linear
constraints (\ref{eq: quadratic con 1}) and (\ref{eq: quadratic con
2}), it has been shown in \cite{Rui09} and \cite{rankone} that an exact
rank-one optimal solution can always be constructed from a non-rank-one
optimal solution. The transformation techniques developed in
\cite{Rui09} and \cite{rankone} can be used to obtain the rank-one
solution. Note that the beamforming solution for the non-reciprocal
channel case is fully centralized, which cannot be implemented in a
partially distributed fashion.

\subsection{Individual-Power Constrained Case}\label{subsec:indivpower} In
the previous subsection, we have discussed the sum-power constrained
case where the non-convex rate maximization problem is converted into a
sequence of convex sum-power minimization problems. In this subsection,
we put a stricter limitation on the relay power by assuming that each
node has its individual power constraint. In this case, following a
similar SDR technique to that in the previous subsection, the
optimization problem with individual power constraints can be cast as
\begin{eqnarray}
    \mathop {\max }\limits_{\mathbf{X}, R_{sum} } {\rm  } & & R_{sum} \label{eq: indiv SDP}\\
    \text{s.t.} & &\frac{P_{S2}\text{tr}(\mathbf{F}_2 \mathbf{X})}{\sigma_{S1}^2+\text{tr}(\mathbf{A_1} \mathbf{X})}\geq \gamma_1, \\
   & &  \frac{P_{S1} \text{tr}(\mathbf{F}_1 \mathbf{X})}{\sigma_{S2}^2+\text{tr}(\mathbf{A_2} \mathbf{X})} \geq \gamma_2, \\
   & & \mathbf{D}_{i,i} \mathbf{X}_{i,i} \leq  P_{R}^{i},~i=1, \cdots, K,  \\
    & & \mathbf{X} \succeq 0,
\end{eqnarray}
where $\mathbf{D}_{i,i}$ and $\mathbf{X}_{i,i}$ are the $i$-th diagonal
elements of $\mathbf{D}$ and $\mathbf{X}$, respectively. The
 transmit power at node $i$ amounts to $\mathbf{D}_{i,i} \mathbf{X}_{i,i}$
and the individual power limit at node $i$ is $P_{R}^{i}$. However, we
cannot translate the above problem into a sequence of power feasibility
problems as given in the last subsection, since we now have $K$
individual power constraints rather than a single sum-power constraint
for the whole relay cluster. Alternatively, we aim at solving a
sequence of the following problem via bi-section search over $r$.
\begin{eqnarray}\label{eq: max X r}
  \mathop {\max }\limits_{\mathbf{X},r}  & & r \\
  \textrm{s.t.} & & \textrm{tr}[(P_{S2}\mathbf{F}_2-\gamma_1\mathbf{A}_1)\mathbf{X}] \geq \gamma_1 \sigma_1^2,  \\
     & & \textrm{tr}[(P_{S1}\mathbf{F}_1-\gamma_2\mathbf{A}_2)\mathbf{X}] \geq \gamma_2 \sigma_2^2,  \\
    & & \mathbf{X}(i,i) \leq   P_{R}^{i}/\mathbf{D} (i,i) ,~i=1,\cdots, K,\\
    & & \mathbf{X} \succeq 0.
\end{eqnarray}
The above problem is convex over $\mathbf{X}$ at each given value of
$r$. Let $r^{*}$ be the maximum value obtained by solving (\ref{eq: max
X r}). For a given value of $r$, we solve the following feasibility
problem
\begin{eqnarray}
    \textrm{Find} & & \mathbf{X} \label{eq:feasibility}\\
       \textrm{s.t.} & & \textrm{tr}[(P_{S2}\mathbf{F}_2-\gamma_1\mathbf{A}_1)\mathbf{X}] \geq \gamma_1 \sigma_1^2, \label{eq: const 1}\\
     & & \textrm{tr}[(P_{S1}\mathbf{F}_1-\gamma_2\mathbf{A}_2)\mathbf{X}] \geq \gamma_2 \sigma_2^2, \label{eq: const 2} \\
     & & \mathbf{X}(i,i) \leq   P_{R}^{i}/\mathbf{D} (i,i) ,~i=1,\cdots, K, \label{eq: const 3}\\
     & & \mathbf{X} \succeq 0.
\end{eqnarray}
If it is feasible, we have $r \leq r^{*}$ and the corresponding rate is
achievable. Otherwise, we have $r> r^{*}$ and the corresponding rate is
not achievable. Based on this observation, we apply bi-section search
over $r$ to solve the problem in (\ref{eq: max X r}), where we solve a
convex feasibility problem of (\ref{eq:feasibility}) at each step. We
start with an interval $[0, r_{max}]$ that contains the optimal value
$r^{*}$ where $r_{max}$ can be obtained in a similar way as that for
the sum-power constrained case, and run the following algorithm.\\
\underline{Algorithm 2}
  \begin{itemize}
  \item Initialize $r_{low}$=0, $r_{up}=r_{max}$.
  \item Repeat
  \begin{enumerate}
    \item Set $r\leftarrow \frac{1}{2}(r_{low}+r_{up})$.
    \item Solve the feasibility problem given by (\ref{eq:feasibility})-(\ref{eq: const 3}) with given $r$.
    \item Update $r$: If the problem is feasible, set $r_{low}= r$; otherwise, $r_{up}= r$.
  \end{enumerate}
  \item Until $r_{up}-r_{low}<\epsilon$. Then $r^{*}=r_{low}$.
\end{itemize}

\subsubsection{Rank-one solution based on randomization}
Similar to the sum-power constrained case, the solution of $\mathbf{X}$
at the end of Algorithm 2, denoted as $\mathbf{X}_{opt}$, may not be
rank-one. However, since there are $K+2$ linear constraints here, we
cannot apply the rank-one decomposition technique in \cite{rankone},
which require the number of linear constraints to be less than or equal
to 3. Fortunately, various techniques have been developed
\cite{randomization} to generate good rank-one approximate solutions to
the original problem. One such efficient approach is based on
randomization \cite{randomization}: using $\mathbf{X}_{opt}$ to
randomly generate a set of candidate weight vectors,
$\{\mathbf{w}_{l}\}$, from which the ``best'' solution for the
beamforming vector $\mathbf{w}$ is selected. There are three ways of
generating $\{\mathbf{w}_{l}\}$ as presented in \cite{randomization}.
In order to satisfy the individual power constraint,  we adopt the
routine named \verb"randB" in \cite{randomization}. Specially, let
$\mathbf{e}_{l}$ be the vector whose elements are independent random
variables uniformly distributed on the unit circle in the complex
plane, i.e., its $i$-th element $[\mathbf{e}_{l}]_i=e^{j\theta_{l,i}}$,
where $\theta_{l,i}$'s are independent and uniformly distributed over
$[0, 2\pi)$. We choose $\mathbf{w}_{l}$ such that its $i$-th element
$[\mathbf{w}_{l}]_{i}=\sqrt{[\mathbf{X}_{opt}]_{ii}}[\mathbf{e}_{l}]_i$.
As we see, $|[\mathbf{w}_{l}]_i|^2=[\mathbf{X}_{opt}]_{ii}$; hence the
individual power constraint can be satisfied.

For each $\mathbf{X}^{(l)}$=$\mathbf{w}_l \mathbf{w}_l^{H}$, we
associate each $\mathbf{w}_l$ with a value $v(\mathbf{w}_l)$,
\begin{eqnarray}\label{eq: violation v}
    v(\mathbf{w}_l) &=&\max\bigg(1-\textrm{tr}[(\frac{P_{S1}\mathbf{F}_1}{\gamma_2\sigma^2_{S2}}-\frac{\mathbf{A}_2}{\sigma^2_{S2}})
    \mathbf{w}_{l}\mathbf{w}_{l}^{H}],\nonumber \\
 & &1-\textrm{tr}[(\frac{P_{S2}\mathbf{F}_2}{\gamma_1\sigma^2_{S1}}-\frac{\mathbf{A}_1}{\sigma^2_{S1}})\mathbf{w}_{l}\mathbf{w}_{l}^{H}]\bigg),
\end{eqnarray}
which reflects how much the constraints are violated. The ``best''
weight vector among the candidate vectors is the one that has the
minimum $v(\mathbf{w}_l)$, i.e.,
\begin{eqnarray}\label{eq: opt w}
    l^{*} &=& \mathop {\arg \min_l} v(\mathbf{w}_l),\\
    \mathbf{w}^*&=&\mathbf{w}_{l^*}.
\end{eqnarray}
\section{Sub-optimal Schemes}\label{sec: subopt}
In this section, we propose some suboptimal schemes with lower
complexity for implementation than the optimal ones established in the
previous sections.
\subsection{Reciprocal Channel Case}
In the reciprocal channel case, at first the transmit phases
$\theta_i$'s at the relays are matched to the channels as
$\theta_i=-(\angle h_{1,i}+\angle h_{2,i})$. Then with the sum-power
constraint, we propose the sub-optimal equal power beamforming where
each relay transmits with equal power. With the individual-power
constraints, we propose the max-power beamforming where each relay
transmits with its maximum power.
\begin{enumerate}
  \item \underline{Equal-power beamforming}:
  All the $K$ relay nodes transmit with the same power $P_R/K$; $\theta_i$'s and $x_i$'s for $i=1,\cdots,K$,
   are given as:
  \begin{eqnarray}
    \theta_i&=&-(\angle h_{1,i}+\angle h_{2,i}),\\
    x_i &=&\sqrt{\frac{P_{R}}{K(P_{S1} |h_{1,i}|^2+P_{S2}
    |h_{2,i}|^2+\sigma_i^2)}}. \label{eq: equal power}
  \end{eqnarray}
  \item \underline{Max-power beamforming}: Each relay transmits
  with its maximum allowable power $P_{R,i}$; $\theta_i$'s and $x_i$'s for $i=1,\cdots,K$, are given as:
  \begin{eqnarray}
  \theta_i&=&-(\angle h_{1,i}+\angle h_{2,i}),\\
    x_i&=&\sqrt{\frac{P_{R,i}}{P_{S1} |h_{1,i}|^2+P_{S2} |h_{2,i}|^2+\sigma_i^2}}.\label{eq: max power}
  \end{eqnarray}
\end{enumerate}
These sub-optimal schemes enjoy implementation simplicity since each
relay only requires the local channel information $h_{1,i}$ and
$h_{2,i}$ to decide the transmit phase and $x_i$.
\subsection{Non-reciprocal Channel Case}
For the non-reciprocal channel case, since the transmit phase cannot be
matched to the two-directional channels simultaneously, we propose a
sub-optimal scheme that greedily chooses the transmit phases.
Specifically, each relay chooses the transmit phase to be either
$\angle{h_{1,i}}+\angle{h_{2,i}^r}$ or
$\angle{h_{2,i}}+\angle{h_{1,i}^r}$, whichever maximizes its own
contribution to the overall SNRs at S1 or S2 without considering any
other relays' contributions, i.e., the transmit phase for each relay is
chosen by following criterion:
 \begin{equation}\label{eq: greedy BF angle}
    \theta_i^*= \arg \max \left(\frac{x_i^2P_{S2}|h_{2,i}h_{1,i}^r e^{j\theta_i}|^2}{\sigma_{S1}^2+x_i^2|h_{1,i}^r|^2\sigma_i^2},
    \frac{x_i^2P_{S1}|h_{1,i}h_{2,i}^r e^{j\theta_i}|^2}{\sigma_{S2}^2+x_i^2|h_{2,i}^r|^2\sigma_i^2}\right),
  \end{equation}
where $x_i$ is the transmit amplitude. To determine $x_i$'s, we propose
equal-power beamforming for the sum-power constraint case and max-power
beamforming for the individual-power constraint case, which are given
in (\ref{eq: equal power}) and (\ref{eq: max power}), respectively.

\section{Numerical Results}\label{sec:results}
In the section, we present numerical results to quantify the achievable
rate region for the two-way relay network with distributed beamforming.
We assume that the relay cluster consists of 5 nodes; the channel
coefficients $h_{1,i}$ and $h_{1,i}^{r}$, $i=1,\cdots, K$, are
independent CSCG variables with distribution $\mathcal{CN}(0,1)$; the
channel coefficients $h_{2,i}$ and $h_{2,i}^{r}$, $i=1,\cdots, K$, are
also independent and distributed as $\mathcal{CN}(0,1)$. The noises at
the relays and source nodes are assumed to have unit variance in the
simulations. We change $\mu$ from 0 to 1 with step $0.1$ and obtain
$11$ Pareto boundary points. For each point, we run 100 channel
realizations to measure the expected performance. We then do convex
hulling over these points.

\begin{figure}
 \begin{centering}
  \includegraphics[width=0.5\textwidth]{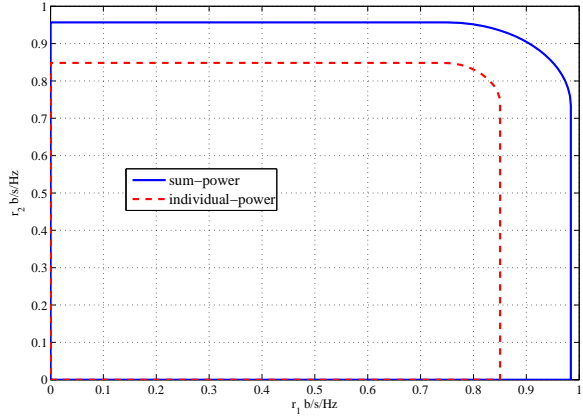}\\
  \caption{Achievable rate regions for the reciprocal channel case, $P_{S1}$=$P_{S2}$=0 dB, $P_R$=10 dB.}\label{fig:regions_rep_fig}
  \end{centering}
\end{figure}
\begin{figure}
 \begin{centering}
  \includegraphics[width=0.5\textwidth]{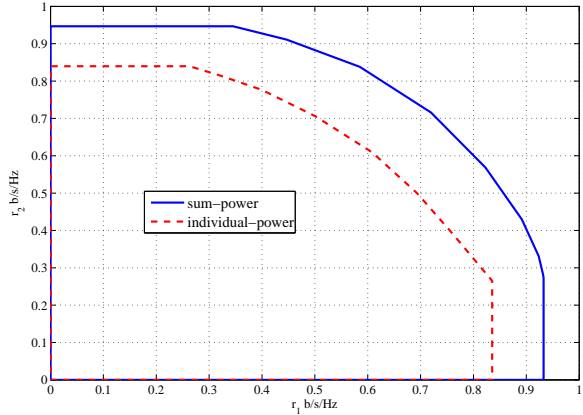}\\
  \caption{Achievable rate regions for the non-reciprocal channel case, $P_{S1}$=$P_{S2}$=0 dB, $\mathbf{P}_R$= [2.5, 3, 0.5, 1, 3] W. }\label{fig:regions_nonrep_fig}
  \end{centering}
\end{figure}

First, we investigate the achievable rate region when channels are
reciprocal, where we set $h_{1,i}=h_{1,i}^r$ and $h_{2,i}=h_{2,i}^r$
for $i=1,\cdots, K$. Fig. \ref{fig:regions_rep_fig} shows the
achievable rate regions with the sum-power constraint and
individual-power constraints, respectively. For the sum-power
constraint case, the relay power $P_R=10$ dB (dB is relative to the
unit noise power) while the transmit powers $P_{S1}=P_{S2}=0$ dB. For
the individual-power constraint case, the relay power constraints are
given as 2.5, 3, 0.5, 1, 3 W, which is summed up to 10 dB. For the
non-reciprocal channel case, Fig. \ref{fig:regions_nonrep_fig} shows
the achievable rate regions with the sum-power and individual-power
constraints, respectively. The powers are the same as the settings in
the reciprocal channel case. We use CVX, a Matlab-based optimization
software \cite{CVX}, to solve the SDP problems. As we see in both Fig.
\ref{fig:regions_rep_fig} and Fig. \ref{fig:regions_nonrep_fig}, due to
the symmetry of the transmit powers and channel statistics, the
achievable rate region $\mathcal{O}$ is symmetric. When $P_{S2}=0$, the
rate pairs collapse to the segment on the horizontal axis, which
corresponds to the achievable rate for a one-way relay network where
only S1 transmits. Moreover, the rate region for the individual-power
constraint case is smaller than that for the sum-power constraint case.
This is quite intuitive since the individual-power constraint is
stricter than the sum-power constraint.
\begin{figure}
 \begin{centering}
  \includegraphics[width=0.5\textwidth]{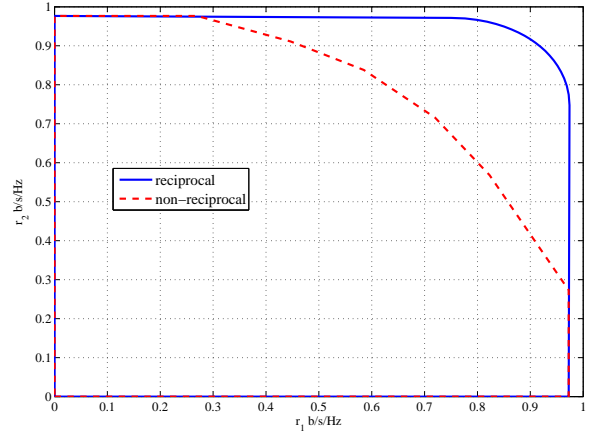}\\
  \caption{Achievable rate regions for the sum-power constraint case, $P_{S1}$=$P_{S2}$=0 dB, $P_R$=10 dB.}\label{fig:regions_sum_fig}
  \end{centering}
\end{figure}

\begin{figure}
 \begin{centering}
  \includegraphics[width=0.5\textwidth]{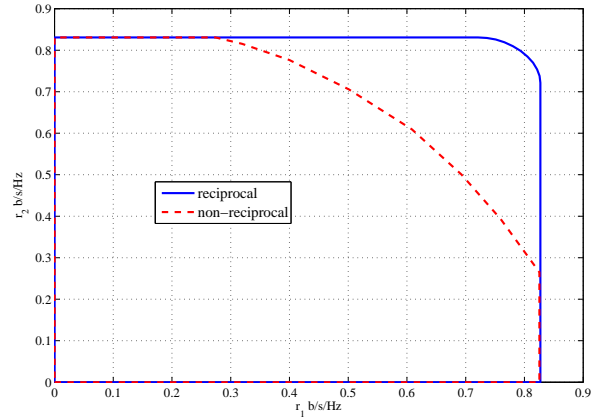}\\
  \caption{Achievable rate regions for the individual-power constraint case,
  $P_{S1}$=$P_{S2}$=0 dB, $\mathbf{P}_R$= [2.5, 3, 0.5, 1, 3] W.}\label{fig:regions_indv_fig}
  \end{centering}
\end{figure}
In Fig. \ref{fig:regions_sum_fig} and Fig. \ref{fig:regions_indv_fig},
we compare the rate regions for the reciprocal and non-reciprocal
channel cases under the same power constraint assumption, where Fig.
\ref{fig:regions_sum_fig} is for the sum-power constraint case and Fig.
\ref{fig:regions_indv_fig} is for the individual-power constraint case.
In both Fig. \ref{fig:regions_sum_fig} and Fig.
\ref{fig:regions_indv_fig}, we can see that the maximum rate for S1 in
the reciprocal channel case is the same as the one in the
non-reciprocal channel case. This is because such a maximum rate is
obtained by optimizing the one-way link from S1 to S2 without
considering the link from S2 to S1. Since the one-way link from S1 to
S2 consists of $\mathbf{h}_1$ and $\mathbf{h}_2^r$, whether
$\mathbf{h}_1^r=\mathbf{h}_1$ or not does not affect the statistics of
the one-way link from S1 to S2. The same argument holds for maximum
rate at S2. We also observe that the rate region for the reciprocal
channel case is larger than that in the non-reciprocal channel case
given the same settings of powers and noises. The reason is that we can
match the beamforming phase to the overall channel phase (i.e., $\angle
\mathbf{w}=\angle \mathbf{h}_1 +\angle \mathbf{h}_2$.) in the
reciprocal channel case, while we are not able to do so in the
non-reciprocal channel case. Therefore, TDD based system is more
favorable in terms of the achievable rate region if the channel
coherence time is larger than one operation period and the
transmit-receive chain calibration \cite{MIMOcom} can be properly done.
Besides the rate region, the amount of information needs to be
broadcast by the control center is significantly different. In the
reciprocal channel case, the control center only needs to broadcast one
scalar at each time slot. However, in the non-reciprocal channel case,
the control center needs to broadcast the beamforming vector, which is
a complex vector of dimension $K$.

\begin{figure}
 \begin{centering}
  \includegraphics[width=0.5\textwidth]{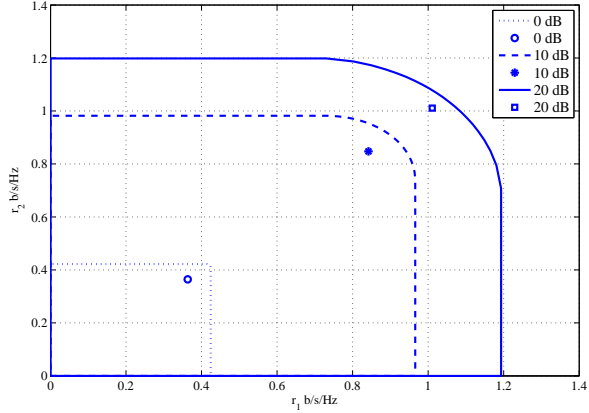}\\
  \caption{Achievable rate regions and equal-power beamforming rates with a sum-power constraint, reciprocal channel,
  $P_{S1}$=$P_{S2}$=0 dB, $P_R$=0, 10, 20 dB.}\label{fig:region vs PR_sum rep}
  \end{centering}
\end{figure}

\begin{figure}
 \begin{centering}
  \includegraphics[width=0.5\textwidth]{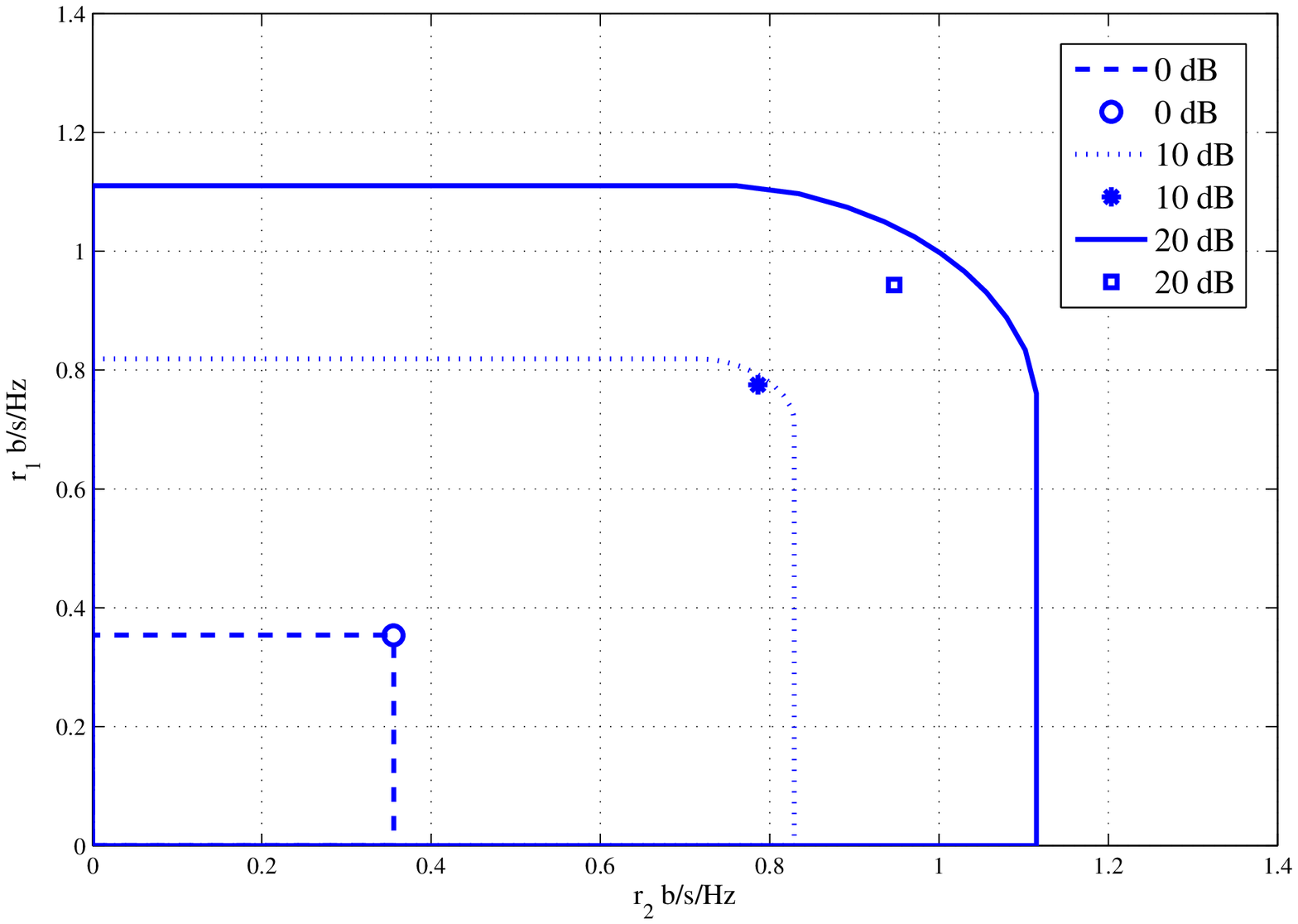}\\
  \caption{Achievable rate regions and max-power beamforming rates with individual-power constraints, reciprocal channel,
  $P_{S1}$=$P_{S2}$=0 dB, total power=0, 10, 20 dB.}\label{fig:region vs PR_indv rep}
  \end{centering}
\end{figure}

\begin{figure}
 \begin{centering}
  \includegraphics[width=0.5\textwidth]{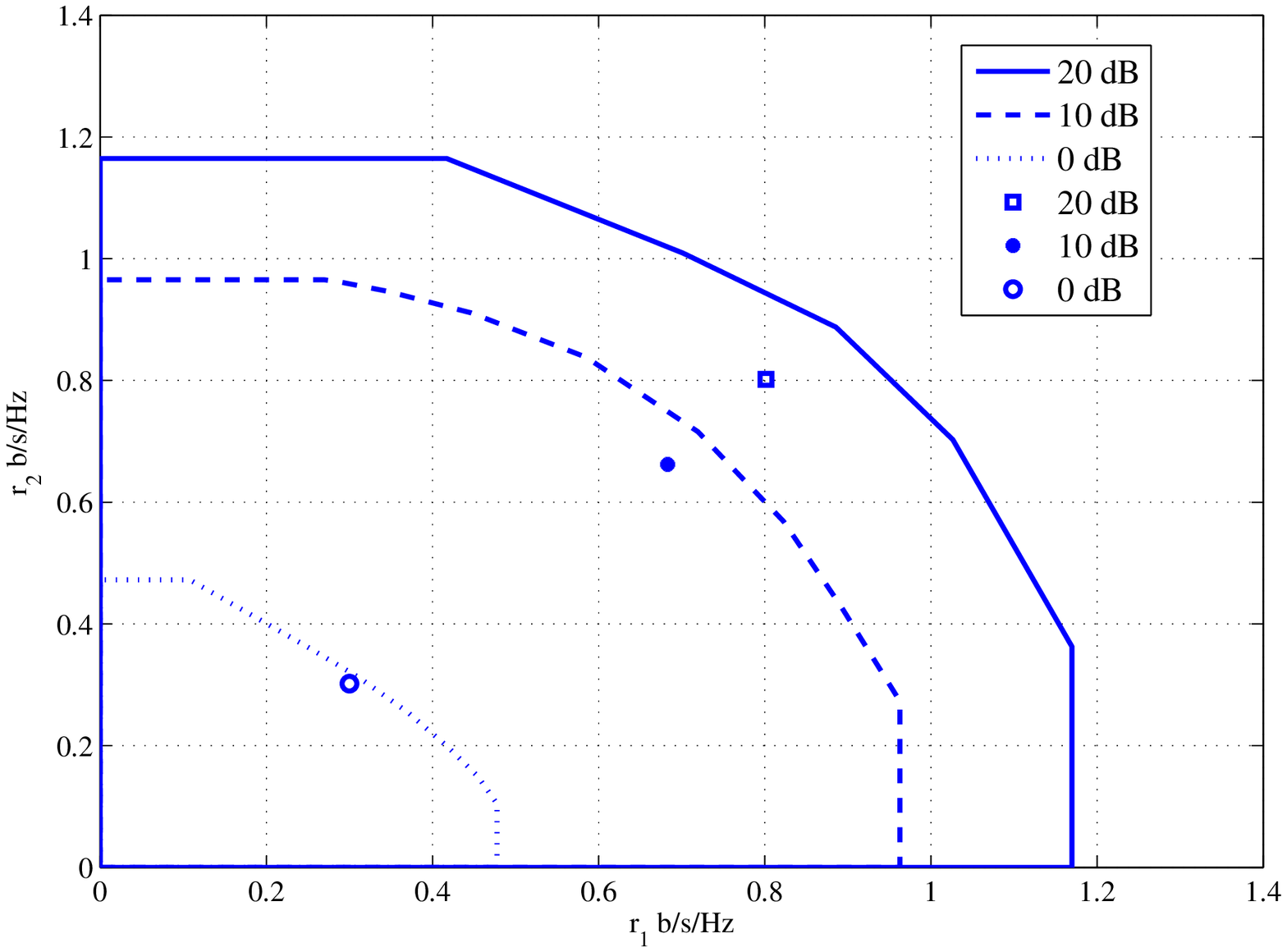}\\
  \caption{Achievable rate regions and equal-power beamforming rates with a sum-power constraint, non-reciprocal channel,
  $P_{S1}$=$P_{S2}$=0 dB, $P_R$=0, 10, 20 dB.}\label{fig:region vs PR_sum nonrep}
  \end{centering}
\end{figure}

\begin{figure}
 \begin{centering}
  \includegraphics[width=0.5\textwidth]{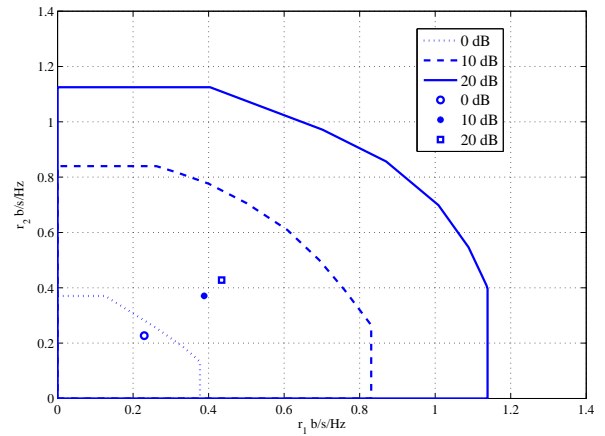}\\
  \caption{Achievable rate regions and maximum-power beamforming rates with individual-power constraints, non-reciprocal channel,
  $P_{S1}$=$P_{S2}$=0 dB, total power=0, 10, 20 dB.}\label{fig:region vs PR_indv nonrep}
  \end{centering}
\end{figure}
At last, we investigate the performance of the sub-optimal schemes in
relative to the maximum achievable rate regions. As we see in Fig.
\ref{fig:region vs PR_sum rep} for the sum-power constraint case with
reciprocal channels, the rate pairs achieved by the equal-power
beamforming scheme, denoted as single points are strictly sub-optimal.
On the contrary, as shown in Fig. \ref{fig:region vs PR_indv rep} for
the individual-power constraint case with reciprocal channels, the rate
pair achieved by max-power beamforming gets closer to the boundary when
the power budget is reduced.\footnote{We set the individual powers
$\mathbf{P}_R= [2.5, 3, 0.5, 1, 3]$ W with total power equal to 10 dB.
When total power is changed to 0 and 20 dB, we scale the vector
proportionally.} In Fig. \ref{fig:region vs PR_sum nonrep} and Fig.
\ref{fig:region vs PR_indv nonrep}, we consider the non-reciprocal
channels and show the performance of equal-power beamforming and
max-power beamforming with greedy phase selection as given in (\ref{eq:
greedy BF angle}). The performance of both equal-power beamforming and
max-power beamforming schemes degrades as $P_R$ increases. Thereby, the
sub-optimal schemes for the non-reciprocal channel case works well only
when $P_R$ is small. In addition, we also evaluate the performance of
equal-power beamforming when channel statistics are asymmetrical, where
we set $h_{1,i}, h_{1,i}^r \sim \mathcal{CN}(0,1)$ and $h_{2,i},
h_{2,i}^r \sim \mathcal{CN}(0,i)$ for $i=1,\cdots,K$. For the
reciprocal channel case, we further set $h_{2,i}=h_{2,i}^r$ and
$h_{2,i}=h_{2,i}^r$. In this case, the equal-power beamforming scheme
as shown in Fig. \ref{fig:region vs PR_sum_rep_asym} and Fig.
\ref{fig:region vs PR_sum_nonrep_asym}, has a larger gap to the
boundary than that in the case where channel statistics are symmetric
(i.e., $h_{1,i}, h_{1,i}^r, h_{2,i}~\text{and}~h_{2,i}^r \sim
\mathcal{CN}(0,1), \text{for}~ i=1,\cdots,K.$), which are shown in Fig.
\ref{fig:region vs PR_sum rep} and Fig. \ref{fig:region vs PR_sum
nonrep} for the reciprocal and non-reciprocal channel cases,
respectively. This is due to the fact that the average transmit powers
for different relays with optimal beamforming should appear non-uniform
if the channel statistics are asymmetric. Therefore, the gain attained
by optimal beamforming becomes more significant in this case.
\begin{figure}
 \begin{centering}
  \includegraphics[width=0.5\textwidth]{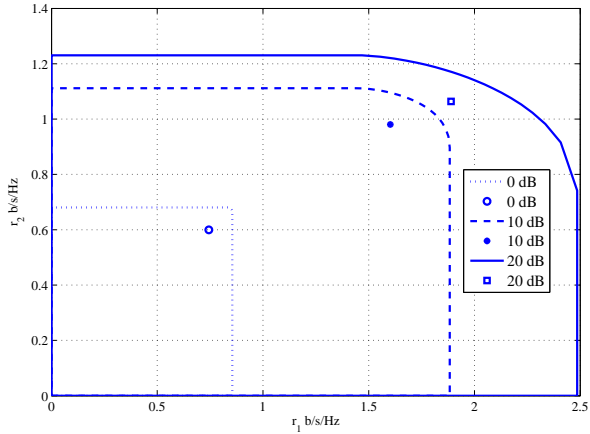}\\
  \caption{Achievable rate regions and equal-power beamforming rates with a sum-power constraint, reciprocal channel,
  $P_{S1}$=$P_{S2}$=0 dB, $P_R$=0, 10, 20 dB.}\label{fig:region vs PR_sum_rep_asym}
  \end{centering}
\end{figure}

\begin{figure}
 \begin{centering}
  \includegraphics[width=0.5\textwidth]{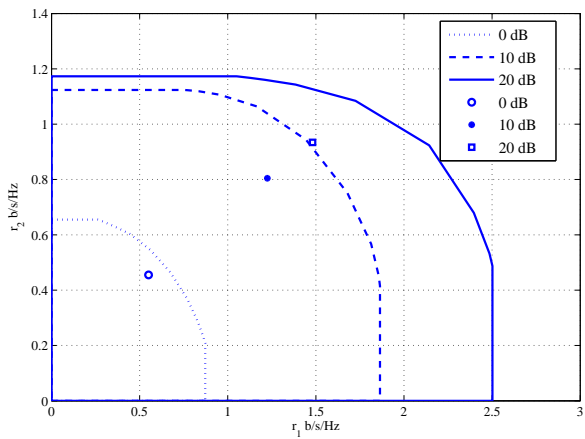}\\
  \caption{Achievable rate regions and equal-power beamforming rates with a sum-power constraint,
  non-reciprocal channel, $P_{S1}$=$P_{S2}$=0 dB, $P_R$=0, 10, 20 dB.}\label{fig:region vs PR_sum_nonrep_asym}
  \end{centering}
\end{figure}

\section{Conclusion}\label{sec: conclusion}
In this paper, we considered the two-way relay networks with
distributed beamforming and investigated the achievable rate region,
which is defined as the convex hull of all achievable rate pairs. We
studied both the reciprocal and non-reciprocal channel cases. In the
reciprocal channel case, we characterized the rate region when the
relay cluster is subject to either a sum-power constraint or
individual-power constraints, respectively. It was shown that we could
characterize the whole achievable rate region via the Pareto-optimal
beamforming vectors obtained from solving a sequence of WSISMin
problems. Furthermore, we derived the closed-form solutions for those
optimal beamforming vectors and consequently proposed partially
distributed algorithms to implement the optimal beamforming, where each
relay node only needs its own local channel information and one global
scalar sent from the control center. For the non-reciprocal channel
case, we used the rate-profile approach to compute the Pareto-optimal
beamforming vectors. When the relay cluster is subject to a sum-power
constraint, we computed the optimal beamforming vector via solving a
sequence of relaxed SDP power minimization problems followed by a
special rank-one reconstruction. When the relay cluster is subject to
individual-power constraints, we solved a sequence of relaxed SDP
feasibility problems and the rank-one solution is obtained by
randomization techniques. From the numerical results, we found that the
achievable rate region is larger in the reciprocal channel case than
that in the non-reciprocal channel case. Hence, TDD-based relaying
scheme is more favorable for the two-way relay network with distributed
beamforming.

\section{Appendices}

\subsection{Proof of Proposition \ref{pp: opt points dual}}\label{proof: pp opt points dual}
\begin{proof}
We will show this by contradiction. Assume $(a,b)\in \mathcal{B}$ but
$\mathcal{U}((a,b))~\overline{\in}~\mathcal{P}$. Then we can find
another point $(c,d)\in \mathcal{R}$ such that $c>1/2\log_2(1+1/a)$ and
$d>1/2\log_2(1+1/b)$. According to the definition of $\mathcal{I}$, the
point $(\frac{1}{2^{2c}-1},\frac{1}{2^{2d}-1})\in \mathcal{I}$. Thus,
there exists a point in $\mathcal{I}$ such that $\frac{1}{2^{2c}-1}< a$
and $\frac{1}{2^{2d}-1}<b$, which contradicts the assumption that
$(a,b)$ is a Pareto optimal point. Hence
$\mathcal{U}(\mathcal{B})\subseteq \mathcal{P}$. The converse that
$\mathcal{U}(\mathcal{B})\supseteq \mathcal{P}$ can also be proven in
the similar way. Therefore, $\mathcal{P}=\mathcal{U}(\mathcal{B})$.
\end{proof}

\subsection{Proof of Proposition \ref{pp: unattainable points unnecessary}}\label{proof: pp unattainable points unnecessary}
\begin{figure*}
\begin{eqnarray}\label{eq: positive 2dtv}
  \frac{dp^2(y)}{dy^2} &=& \frac{d p'(y)/dx}{dy/dx}\nonumber\\
    &=&-\frac{x+x^2}{\ln2}\left[\frac{q^{''}(x)(x+x^2)}{q(x)+q(x)^2}+q^{'}(x) \left(\frac{(1+2x)(q(x)+q(x)^2)-(q^{'}(x)+2q(x)q^{'}(x))(x+x^2)}{q(x)+q(x)^2} \right)\right]> 0.
\end{eqnarray}
\end{figure*}
In order to prove Proposition \ref{pp: unattainable points
unnecessary}, we first introduce the following two lemmas.
\begin{lemma}\label{lma: concave mapping}
Suppose $q(x)$ is a positive, decreasing, and linear function with $x
> 0$. The bijective mapping $\mathcal{U}$ maps $(x,q(x))$ to
$(y,p(y))$; then $p(y)$ is a non-negative, decreasing, and convex
function.
\end{lemma}

\begin{proof}
Let $y=\log_2(1+1/x)$ and $p(y)=\log_2(1+1/q(x))$ be an implicit
function of $y$, where $x > 0$. Since $q(x)$ is positive, decreasing,
and linear, we have $q(x)>0$, $q'(x)<0$, $q^{''}(x)=0$, and hence
$p(y)\geq0$. The first-order derivative of $p(y)$ is
\begin{eqnarray*}
  p'(y) &=& \frac{d p(y)/dx}{dy/dx} \\
   &=& q'(x)\frac{x+x^2}{q(x)+q(x)^2}<0.
\end{eqnarray*}
The second-order derivative is given by (\ref{eq: positive 2dtv}),
which is positive.

Thus, $p(y)$ is a convex function of $y$.
\end{proof}
\begin{figure}[htbp]
\centering
 \subfigure[A straight line in $\mathcal{I}$.]
 {\includegraphics[width=0.2\textwidth]{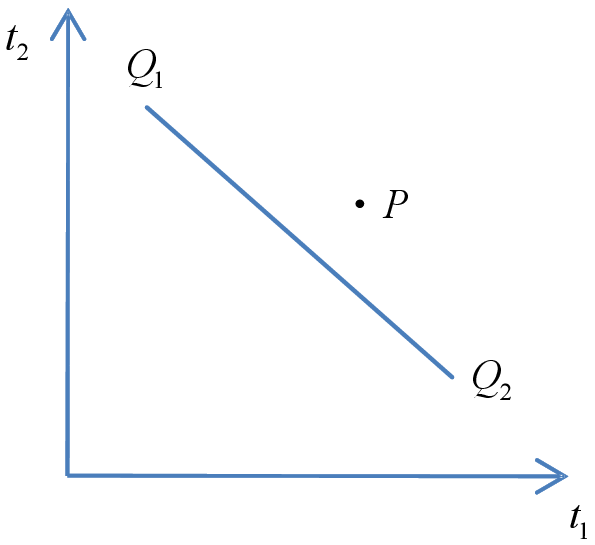}
\label{fig: line2cvx mapping 1}}
 \subfigure[The image of $\overline{Q'_{1}Q'_{2}}$ in $\mathcal{R}$.]
{\includegraphics[width=0.2\textwidth]{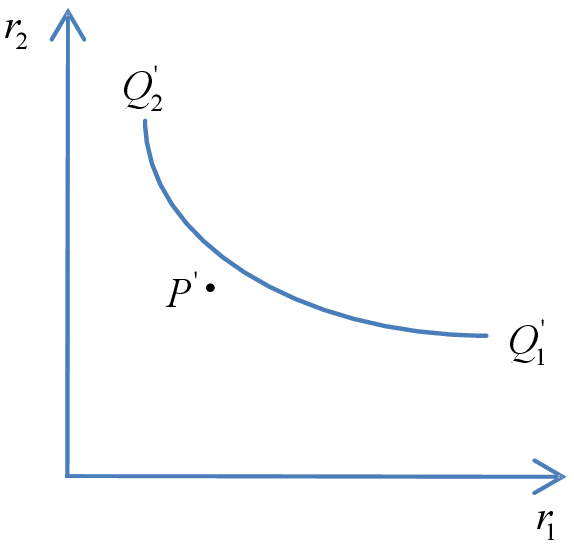} \label{fig:
line2cvx mapping 2}}
 \caption{Illustration of Lemma \ref{lma: concave mapping}. Mapping a straight line in $\mathcal{I}$ to a convex curve in $\mathcal{R}$.}\label{fig: line2cvx mapping}
\end{figure}
According to the above lemma, the line segment $\overline{Q_{1}Q_{2}}$
in Fig. \ref{fig: line2cvx mapping 1} is mapped to a convex curve
$\widehat{Q_{1}^{'}Q_{2}^{'}}$ in Fig. \ref{fig: line2cvx mapping 2} by
$\mathcal{U}$. In addition, it is easy to see that
$\overline{Q_{1}Q_{2}}+\mathbf{K}\mapsto
\widehat{Q_{1}^{'}Q_{2}^{'}}-\mathbf{K}$, i.e., any point above
$\overline{Q_{1}Q_{2}}$ (for example, $P$ in Fig. \ref{fig: line2cvx
mapping 1}) will be mapped to be a point below
$\widehat{Q_{1}^{'}Q_{2}^{'}}$ (i.e., $P^{'}$ in Fig. \ref{fig:
line2cvx mapping 2}).

\begin{lemma}\label{lma: above line}
Let a point $(q_1,q_2)\in \textbf{bd}(\mathcal{I})\setminus
\mathcal{S}(\mu, \mathcal{I})$.  If $q_1=\lambda t_1+\bar{\lambda}s_1$,
where $(t_1,t_2),(s_1,s_2)\in \mathcal{S}(\mu, \mathcal{I})$ and
$(t_1,t_2)\neq(s_1,s_2)$, we have $q_2
> \lambda t_2+\bar{\lambda}s_2$, i.e., the point $(q_1,q_2)$ is above
the line segment connecting $(t_1,t_2)$ and $(s_1,s_2)$.
\end{lemma}

\begin{proof}
We show this by contradiction. Suppose $\mathcal{S}(\mu, \mathcal{I})$
has more than one elements for a given $\mu$, such that $(t_1,t_2),
(s_1,s_2)\in \mathcal{S}(\mu, \mathcal{I})$, and $(t_1,t_2)\neq
(s_1,s_2)$. According to the definition of $\mathcal{S}(\mu,
\mathcal{I})$ given by (\ref{eq:set mu WS}), we have $\mu t_1+
\bar{\mu} t_2=\mu s_1+ \bar{\mu} s_2=m$, where $m$ is the minimum value
of the weighted sum for a given $\mu$ over all points in $\mathcal{I}$.
If $q_1=\lambda t_1+\bar{\lambda}s_1$ and $q_2 \leq \lambda
t_2+\bar{\lambda}s_2$, we have
\begin{eqnarray}
  \mu q_1+\bar{\mu}q_2 &\leq& \mu(\lambda
    t_1+\bar{\lambda}s_1)+\bar{\mu}(\lambda
    t_2+\bar{\lambda}s_2) \\
   &=& \lambda (\mu t_1+\bar{\mu}t_2)+\bar{\lambda} (\mu s_1+\bar{\mu}s_2) \\
   &=& m.
\end{eqnarray}
If $\mu q_1+\bar{\mu}q_2< m$, it contradicts that $m$ is the minimum
value of the weighted sum for the given $\mu$; If $\mu
q_1+\bar{\mu}q_2= m$, it contradicts that $(q_1,q_2)$ is not in
$\mathcal{S}(\mu, \mathcal{I})$. Therefore, the lemma holds.
\end{proof}

According to the above lemma, for a given $\mu$, if
$\mathcal{S}(\mu,\mathcal{I})$ has more than one elements,  the set of
boundary points $\{(q_1,q_2): (q_1,q_2)\in
\textbf{bd}(\mathcal{I})\setminus \mathcal{S}(\mu,\mathcal{I}),
s_1<q_1<t_1, (s_1, s_2), (t_1, t_2)\in \mathcal{S}(\mu,\mathcal{I})\}$
must be above the line segment connecting $(s_1, s_2)$ and $(t_1,
t_2)$; and hence are not attainable by solving WSISMin. This is true
for all $\mu$'s; hence \emph{if a boundary point is not attainable by
solving WSISMin, it must be above a line segment connecting two
particular points in $\mathcal{\mathcal{S}(\mu,\mathcal{I})}$ for some
$\mu$}. With the above two lemmas, we are ready to prove Proposition
\ref{pp: unattainable points unnecessary} as follows.

\underline{Proof of Proposition \ref{pp: unattainable points
unnecessary}}:
\begin{proof}
\begin{figure}[htbp]
\centering
  \subfigure[Inverse-SNR region]
{\includegraphics[width=0.2\textwidth]{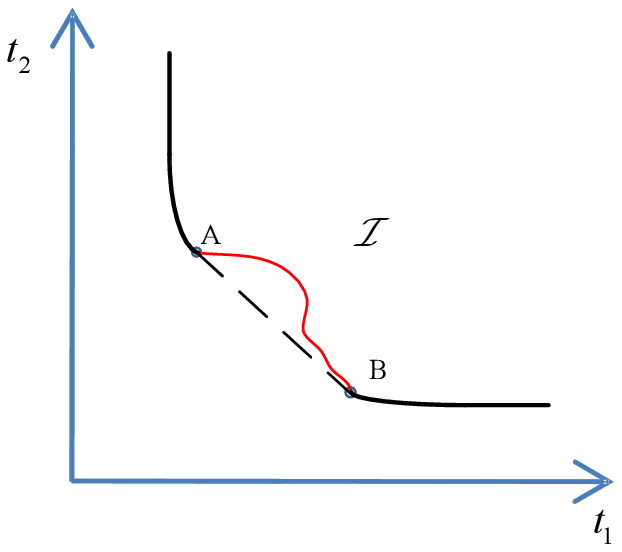}
\label{fig: inverseSNR_curves}} \subfigure[Rate region]
 {\includegraphics[width=0.2\textwidth]{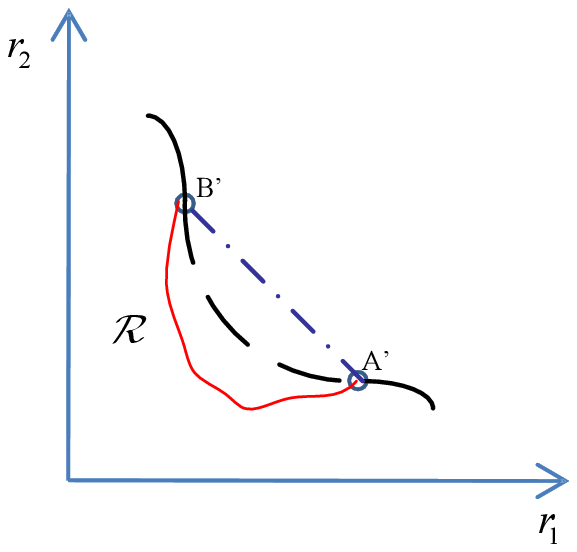}
\label{fig: rate_curves}}
 \caption{Inverse-SNR region and corresponding rate and  region.}\label{fig: mapping curves}
\end{figure}
First we define
\begin{equation}
    \Delta=\{\mu: \mathcal{S}(\mathcal{\mu, I}) \text{ has more than one elements}\},
\end{equation}
and let $l_{\mu}$ be the line segment (e.g., $\overline{AB}$  in Fig.
\ref{fig: inverseSNR_curves})
 with two end points from $\mathcal{S}(\mathcal{\mu, I})$  for $\mu\in \Delta$
(e.g.,  points $A$ and $B$  in \ref{fig: inverseSNR_curves}). According
to Lemma \ref{lma: above line}, the boundary points that are not
attainable by solving WSISMin, denoted as
$\textbf{bd}(\mathcal{I})\setminus\mathcal{S}(\mathcal{I})$ (here
referring to curve $\widehat{AB}$ in Fig. \ref{fig:
inverseSNR_curves}), must be above $l_{\mu}$'s, i.e.,
$\textbf{bd}(\mathcal{I})\setminus\mathcal{S}(\mathcal{I}) \subseteq
\bigcup_{\mu\in \Delta} (l_{\mu}+\mathbf{K})$; and it follows that
$\mathcal{U}(\textbf{bd}(\mathcal{I})\setminus
\mathcal{S}(\mathcal{I})) \subseteq \mathcal{U}(\bigcup_{\mu\in \Delta}
(l_{\mu}+\mathbf{K})) $. According to Lemma \ref{lma: concave mapping},
$\mathcal{U}(\bigcup_{\mu\in \Delta} (l_{\mu}+\mathbf{K}))\subseteq
\bigcup_{\mu\in \Delta} ( \mathcal{U}(l_{\mu})-\mathbf{K})$, where
$\mathcal{U}(l_{\mu})$ is a convex curve (e.g., here
$\mathcal{U}(l_{\mu})$ refers to the dashed convex curve
$\widehat{A'B'}$ in Fig. \ref{fig: rate_curves}). Let $\tilde{l}_{\mu}$
be a line segment (i.e., the dot-dashed line segment $\overline{A'B'}$
in Fig. \ref{fig: rate_curves}) that connects the two end points of the
convex curve $\mathcal{U}(l_{\mu})$. Due to the convexity of
$\mathcal{U}(l_{\mu})$, we have $\mathcal{U}(l_{\mu})-\mathbf{K}
\subseteq \tilde{l}_{\mu}-\mathbf{K}$ and hence
$\mathcal{U}(\textbf{bd}(\mathcal{I})\setminus\mathcal{S}(\mathcal{I}))
\subseteq \bigcup_{\mu\in \Delta} ( \tilde{l}_{\mu}-\mathbf{K})$.
Notice $\bigcup_{\mu\in \Delta} ( \tilde{l}_{\mu}+\mathbf{K})$ is
sufficient for constructing $\mathcal{O}$ by convex hulling. Therefore,
$\mathcal{U}(\textbf{bd}(\mathcal{I})\setminus\mathcal{S}(\mathcal{I}))$
or $\textbf{bd}(\mathcal{I})\setminus\mathcal{S}(\mathcal{I})$ is not
necessary for constructing $\mathcal{O}$. Since $\mathcal{B} \subseteq
\textbf{bd}(\mathcal{I})$, the set
$\mathcal{B}\setminus\mathcal{S}(\mathcal{I})$ is also not necessary
for constructing $\mathcal{O}$.
\end{proof}

\label{sec:ref}

\bibliographystyle{IEEEtran}
\bibliography{strings}

\end{document}